\documentclass[letterpaper, 10 pt, conference]{ieeeconf}

\IEEEoverridecommandlockouts
\usepackage[english]{babel}

\usepackage{amssymb}
\usepackage{mathtools}
\usepackage{graphicx}
\usepackage{mathrsfs}

\newtheorem{theorem}{Theorem}[section]
\newtheorem{assumption}{Assumption}[section]
\newtheorem{lemma}[theorem]{Lemma}
\newtheorem{coro}[theorem]{Corollary}
\newtheorem{definition}{Definition}
\newtheorem{remark}{Remark}
\newtheorem{example}{Example}[section]

\newcommand{\z}{{\boldsymbol z}}
\newcommand{\y}{{\boldsymbol y}}
\newcommand{\w}{{\boldsymbol w}}
\newcommand{\bu}{{\boldsymbol u}}
\newcommand{\LL}{\mathcal{L}}
\newcommand{\WW}{\mathcal{W}}
\newcommand{\LLD}{\mathscr{L}}
\newcommand{\RR}{\mathscr{R}}
\newcommand{\EE}{\mathscr{E}}
\newcommand{\DD}{\mathscr{D}}
\newcommand{\cS}{\mathcal{S}}
\newcommand{\sat}{{\mathrm{sat}}}
\newcommand{\sign}{{\mathrm{sign}}}
\newcommand{\diag}{{\mathrm{diag}}}

\title{\LARGE \bf
Local Input-to-State Stability for 
Consensus in the Presence of
Intermittent
Communication and Input Saturation
}

\author{Thales C. Silva and M. Ani Hsieh
\thanks{
This work was supported by ARL DCIST CRA W911NF-17-2-0181 and Office of Naval Research (ONR) Award No. N00014-19-1-2253.
}
\thanks{
The authors are with the
General Robotics, Actuation, 
Sensing, and Perception (GRASP) Laboratory at University of Pennsylvania, Philadelphia, PA,
19104, USA
        {\tt\small \{scthales, m.hsieh\}@seas.upenn.edu}.
        Accepted for publication in the Proceedings of the IEEE Conference on Decision and Control (CDC), 2026.
Final version: https://doi.org/10.1109/CDC49753.2023.10383364
        }
}

\begin{document}

\maketitle
\thispagestyle{empty}
\pagestyle{empty}

\begin{abstract}
This paper addresses the problem of reaching
 consensus under input saturation and intermittent communication, which can hinder the convergence of the system.
 We propose a method that translates the consensus into an equivalent stability problem. Then, we compute 
 bounded sets that enclose the initial conditions and the evolution of trajectories leading to local input-to-state stability for systems interconnected over directed intermittent topologies. Our contributions include sufficient conditions for stability and stabilization of multi-agent systems under intermittent interactions and saturating inputs, with the ability to evaluate disturbance tolerance and rejection 
based on the regions that enclose the 
system's
trajectories. We define disturbance rejection in terms of the $\LLD_2$ gain, and formulate stability and controller design conditions as convex optimization problems.
Our method enable the maximization of regions that ensure local input-to-state stability, we provide numerical examples highlighting the trade-offs between mean frequency of intermittent interactions, disturbance energy, and convergence region size.
\end{abstract}

\section{Introduction}
Multi-agent systems consist of multiple agents that can interact, coordinate, cooperate, and/or compete with one another to perform complex tasks (see \cite{Li2019,Chen2021} for recent advances).  These systems have  
wide applicability 
across different disciplines,
such as 
robotic swarms \cite{Yu2021},
environmental monitoring 
\cite{Lynch2008},
and surveillance \cite{Yu2020}.
In many applications, 
the agents need to reach agreement on the value of a variable of interest, 
resulting in the multi-agent system 
achieving {\it consensus}.
There is a vast and increasing literature on the consensus problem \cite{thales2021,thales2023,Chen2019b,Olfati-Saber2007}.

In distributed consensus, agents within a
team must interact with each other, and thus the majority of the literature assumes
all-time network connectivity between agents.
Nevertheless, 
all-time
connectivity, 
either by control or by assumption,
can be overly conservative,
and even
impractical, especially for monitoring large-scale environments like the ocean
\cite{XiYuDARS},
underground tunnels
\cite{Chaimowicz2005},
or urban cities where wireless signal can be severely attenuated due to obstacle occlusions and other environmental factors.

In real-world applications, multi-agent systems must contend with intermittent network connectivity, input saturation, modeling errors, and exogeneous disturbances, all of which can lead to significant performance degradation of the team.  In terms of consensus, these conditions can give rise to deteriorating convergence rate, creation of limit cycles, and even instability \cite{Tarbouriech2011,hu2001control}.  Existing work in multi-agent systems mostly address these challenges individually.  Works focused on the effects intermittent network connectivity on multi-agent systems include \cite{Song2017,Mu2018,Poulsen2019}.  Unfortunately, these strategies are only developed for individual agents with either first- or second-order dynamics, with some exceptions \cite{Gao2021,Lin2021}, and thus do not generalize to systems composed of agents with more complex dynamics.  Literature that focus on input saturation effects include \cite{Wei2011,Zhao2016,Zhao2017,Guan2020}.  However, these works either assumes directed \cite{Wei2011} or strongly connected \cite{Zhao2016,Zhao2017} network topologies.  And while \cite{Guan2020} proposed an adaptive fault-tolerant control in the presence of $\LL_2$-limited disturbances and
input saturation, all these works assume \textit{asymptotically null controllable agents}, {\it i.e.}, open-loop
linear systems with no poles on the right-half plane in the continuous-time domain.

An objective of this work is to investigate the impact of intermittent communication topologies on a multi-agent system's ability to achieve consensus.  In particular, we are interested in systems composed of individuals with more general dynamics models subject to constraints and exogenous disturbances.  The current limitations in existing literature is partly due to the fact that \cite{Schmitendorf1980,Sussmann1994} has shown that \textit{global} convergence of linear systems with saturating inputs can  only be achieved by {asymptotically null controllable systems.}  As such, the system's dynamics needs to be open-loop stable, in the Lyapunov sense, due to saturation.
Hence, global asymptotic convergence in these systems can only happen under overly conservative conditions.  To overcome this limitation, it is necessary to characterize regions for initial conditions from which the convergence can be guaranteed, especially when the objective is to design stabilizing feedback strategies for open-loop {\it unstable} linear systems.  From a consensus perspective, this means that it is critical to estimate relative regions that encompass the agents' states from which the group's convergence can be achieved. This is a fundamentally more general and challenging problem than assuming asymptotic null controllability.


Towards this end, we propose a methodology capable of characterizing local stability
for the general case of potentially open-loop unstable
agents ({\it e.g.}, agents with poles of the open-loop system on the right-half plane), subject to
stochastic nonsynchronous link formations, unlike most existing work \cite{Song2017,Mu2018,Poulsen2019,Gao2021,Lin2021}.  To the best of the authors' knowledge,
the combined impact of saturating inputs and intermittent communications have only been addressed in
\cite{Fan2019,Xie2020,Wang2015} but only for asymptotically null controllable agents and assume communication link activation and deactivation happen in a synchronized fashion \cite{Fan2019,Xie2020,Wang2015}. 
%
We 
consider the synthesis of distributed stabilizing controllers for multi-agent systems subject to saturating inputs
over time-varying stochastic communication topologies, with agents subject to disturbances limited in energy and non-zero initial conditions.  
 This is a particularly challenging task since saturating inputs require the characterization of regions that guarantee the system convergence for stability \cite{Tarbouriech2011}, while disturbances and  intermittent communication among agents directly affects the structure of these regions. 
We propose a convex approach to address this problem and ensure that trajectories starting within a specific set remain inside a larger outer set. Disturbance rejection is then defined as the measurement of relative sizes of these two sets for whenever the initial conditions of the network are different from zero, building upon previous work on single systems \cite{hu2001control}.

\subsubsection*{\bf Notation}

For a matrix
$M$, $M>0$ $(M\geq0)$
denotes that $M$ is
positive definite (semi-definite).
$M'$ stands for the transpose
of $M$.
The $i$th row (element) of the
matrix (vector) $N$ is denoted by
$N_{(i)}$, while $N_{(i)}'$
 represents the transpose 
 of the $i$th row
of $N$.
 $\diag(M_1,\dots,M_k)$
represents a block diagonal
matrix with matrices
$M_1,\dots,M_k$ in the main
diagonal. Transposed entries
in a symmetric matrix
are represented by $*$.
The Kronecker product
of matrices $M_1$ and
$M_2$ is
denoted by $M_1\otimes M_2$.

\section{Preliminaries and Problem Formulation}
\label{sec:2}
This section introduces the representation for agent interaction and the control model for closed-loop agents, including the transformation of consensus into a stability problem.

\subsection{Algebraic Graph Theory}
A graph is denoted by 
$\mathcal{G}=(\mathcal{V},\mathcal{E})$, in which 
$\mathcal{V}=\{1,...,N\}$ is a set of $N$ vertices and
$\mathcal{E}\subseteq \mathcal{V}\times \mathcal{V}$ is 
a set of directed edges. Each element of the edge set,
$e_{ij}$, represents a directed edge from $i$ to $j$ if
$(i,j)\in \mathcal{E}$. A graph $\mathcal{G}$ is called
{\it undirected graph} if $(i,j)\in \mathcal{E}
\iff (j,i)\in \mathcal{E}$ and 
{\it directed graph} if the equivalence does not hold.
In this work, we are interested in \textit{directed graphs}.
The set of agents that have the $i$th
agent as the child vertex is 
$\mathcal{N}_i=\{j\in \mathcal{V} : (j,i)\in\mathcal{E}\}$ and we call it
neighborhood of the $i$th vertex, and
we call neighbor of $i$ an agent  $j$ that
belongs to the neighborhood $\mathcal{N}_i$.
The adjacency matrix
$\mathcal{A}=[a_{ij}]$ associated with the graph $\mathcal{G}$ is
defined by
\begin{IEEEeqnarray*}{l}
a_{ij} = \left\{
\begin{matrix*}[l]
0, &\text{if } i=j \text{ or }  \nexists ~ (j,i)\in\mathcal{E}\\
1, & \text{if } (j,i)\in\mathcal{E}.
\end{matrix*}
\right.
\end{IEEEeqnarray*}
We can represent the graph
$\mathcal{G}$ through the 
Laplacian matrix, defined by
$
\mathcal{L} = \mathcal{D} - \mathcal{A},
$
where 
the diagonal
degree matrix, $\mathcal{D}$,
has elements
$d_{ii}=\sum_{j=1}^N a_{ij}$.
\subsection{Dynamical Network}
Consider $N$ agents with the following open-loop
dynamics
\small
\begin{IEEEeqnarray}{l}
\label{eq:MA}
{\dot x}_i(t) = A x_i(t) + B\sat (u_i(t)) + Dw_i(t),
\text{ for } i=1,...,N,~
\end{IEEEeqnarray}
\normalsize
where $x_i(t)\in \mathbb{R}^m$ is the state
variable of the $i$th agent, $u_i(t)\in \mathbb{R}^p$
is its input, $w_i(t)\in \mathbb{R}^q$ is an exogenous signal,
and $A$, $B$, and $D$ are the system 
matrices with appropriate dimensions. The function
$\sat(\cdot)$ is given by, 
\small
\begin{IEEEeqnarray*}{l}
\sat(u_i(t))=[\sat(u_{i(1)}(t))~\cdots~
\sat(u_{i(p)}(t))]', \\
\sat(u_{i(r)}(t)) 
= \sign(u_{i(r)}(t))\min(u_{\max},\lvert u_{i(r)}(t)\lvert ),
\end{IEEEeqnarray*}
\normalsize
for $r=1,...,p$, where the scalar $u_{\max}$ is 
the limit of the actuator.
This study investigates the
impact of saturation on networked systems
subject to stochastic
communication topologies and
input disturbances.
The time-varying interconnections are modeled by
a dependency of the elements of the communication
graph 
in a continuous-time Markov process 
$\{\theta_t:t\geq0\}$ over a finite set of states
$\mathcal{S}=\{1,\dots,s\}$, where the random variable
$\theta_t$ is the state of the Markov chain at time
$t$. This dependency is 
represented by $a_{ij}(\theta_t)$.
We consider a distributed control setting and adopt the following variant of the conventional consensus protocol \cite{Olfati-Saber2002}:
\small
\begin{IEEEeqnarray}{l}
\label{eq:input}
u_i(t) =
- \sum_{j=1}^N a_{ij}(\theta_t)K\big(x_i(t)-x_j(t)\big),
\end{IEEEeqnarray}
\normalsize
in which $K\in\mathbb{R}^{p\times m}$ is a
constant gain matrix.
The value
$a_{ij}(k)=1$ if, and only if, 
$(j,i)\in\mathcal{E}$ when $\theta_t = k$, for 
$k\in \mathcal{S}$.
An implicit assumption on $\theta_t$ in
equation (\ref{eq:input}) is that it models
accurately enough intermittent interactions in the
networked system, these 
parameters
can be computed using identification techniques \cite{Bremaud2020}.
In addition, 
 we work under the following assumption:
\begin{assumption}
\label{assump:21}
The union of graphs,
$\bigcup_{\theta_t\in \mathcal{S}}\mathcal{G}(\theta_t)$,
associated with the time-varying communication topology of the network has a directed spanning tree,
\textit{i.e.,} there is a node with a directed path to every other node in the network.
\end{assumption}

The Markov process 
is
defined by a time-varying transition rate
matrix $\Pi(t)=\left[
\pi_{ij}(t)\right], ~ i,j\in \mathcal{S}$,
which evolution is
described by the following infinitesimal 
generator,
\small
\begin{IEEEeqnarray}{l}
\mathrm{Pr}\{\theta_{t+\Delta} = j \lvert \theta_t=i\}
= \left\{
\begin{matrix*}[l]
\pi_{ij}(t)\Delta + o(\Delta), & i\neq j,\\
1 + \pi_{ii}(t)\Delta + o(\Delta), & i=j,
\end{matrix*}
\right.
\end{IEEEeqnarray}
\normalsize
where
$\mathrm{Pr}\{\theta_{t+\Delta} = j
\lvert \theta_t=i\}$ stands for the probability
of the next state be $\theta_{t+\Delta} = j$
given that the current state is $\theta_t=i$,
$\Delta > 0$ and $\lim_{\Delta \rightarrow 0}\frac{o(\Delta)}{\Delta}=0$,
$\pi_{ij}(t)\geq0$ for $i\neq j$ is the
transition rate from state $i$ to $j$, and
$\pi_{ii}(t) = -\sum_{j=1,j\neq i}^s \pi_{ij}(t)$.
In addition, the probability distribution of the
Markov chain is denoted by
$\mu = \{\mu_1,\dots,\mu_s\}$.
The time-varying transition matrix
$\Pi(t)$ is not precisely known,
but belonging to a
polytopic uncertainty domain denoted by
\small
\begin{IEEEeqnarray}{l}
\label{eq:polytope}
\Pi(\alpha(t)) = \sum_{i=1}^r \alpha_i \Pi^i, 
~\alpha_i \in \Xi_r,
\end{IEEEeqnarray}
\normalsize
where the unit simplex $\Xi_r$ is defined as
\small
$\Xi_r=\big\{{\boldsymbol \alpha}\in \mathbb{R}^r
:0\leq \alpha_i \leq 1, ~\sum_{i=1}^r\alpha_i=1
\big\},
$
\normalsize
in which $r$ is the number of
vertices of the polytope which are given by
known matrices $\Pi^i$ with elements
 $\pi_{jq}^i$, for $i\in\{1,...,r\}.$
 This approach allows us to robustly tackle the problem of time-varying
 and uncertain switching topologies. 

The $\LLD_2$
gain 
is traditionally used
to evaluate the disturbance rejection
of a system \cite{KHALIL2002ab}.\
However, for the multi-agent
system \eqref{eq:MA}, the $\LLD_2$ gain might not be well-defined, since 
sufficiently large disturbances
might lead to unbounded states 
 \cite{Guan2020,DeSouza2018}. 
To address this,
we focus on disturbances whose energy
is bounded, belonging to the following
set:
\small
\begin{IEEEeqnarray*}{l}
\WW=\Bigg\{w_i\in \mathbb{R}^q: 
\int_0^{\infty}w_i(\tau)'w_i(\tau)d\tau < \rho
\Bigg\},
\end{IEEEeqnarray*}
\normalsize
where $\rho$ is a positive constant 
representing the energy bound of the disturbance.
Our goal is to find the largest $\rho$
for which the trajectories of the networked agents remain bounded in the presence of switching topologies and input saturation,
by seeking the $\LLD_2$ gain as the ratio between a bounded output of interest and input limited in energy.
We provide a formal definition in Sec. \ref{subsubsec:estimationL2}.\

Now, formalize the
problem
investigated in this work:
\begin{definition}
\label{def:consensus}
Under stochastic switching topology, the
networked system (\ref{eq:MA}) in closed-loop
with (\ref{eq:input}), for 
all $i\in\mathcal{V}$, reaches
the mean-square consensus if, for all $i\neq j$,
$\lim_{t\rightarrow \infty} \mathbb{E}\big( \lVert x_i(t)
-x_j(t)\lVert^2\big)\rightarrow 0$ holds in the mean square
sense, for any initial distribution $\mu$.
\end{definition}
Here, $\mathbb{E}(\cdot)$ stands for the mathematical
expectation.


\subsection{Mean-Square consensus as a stability problem}

We translate the consensus problem into a stability analysis and derive sufficient conditions that guarantee the stabilization of the equivalent system.
Consider the following disagreement transformation
 \cite{YuanGongSun2009} and its equivalent stacked form,
 \small
\begin{IEEEeqnarray}{lll}
z_i(t)~&=&~x_1(t)-x_{i+1},~\text{for}~i=1,...,N-1,
\nonumber \\
\label{eq:disagreement}
\boldsymbol{z}(t)~&=&~(U\otimes I_m){\boldsymbol x}(t),
\end{IEEEeqnarray}
\normalsize
where $\otimes$ denotes the Kronecker
product, $I_m$ is an identity matrix,
${\boldsymbol z}(t) = [z_1(t)' ~\cdots~ 
z_{N-1}(t)']'$, 
${\boldsymbol x}(t) = [x_1(t)' ~\cdots~
~x_{N}(t)']'$,  and 
$U=[{\boldsymbol 1}_{N-1} ~ -I_{N-1}]$ with
${\boldsymbol 1}_{N-1}$ being an 
all-ones vector
of size $N-1$. The transformation from the
disagreement variables ${\boldsymbol z}(t)$ back to
the stacked agent states can be computed by
\small
\begin{IEEEeqnarray}{l}
\label{eq:MA_back_stacked}
\boldsymbol{x}(t) =
{\boldsymbol 1}_{N}\otimes {x}_{1}(t) + 
(W\otimes I_m){\boldsymbol z}(t),
\end{IEEEeqnarray}
\normalsize
in which 
$ 
W = \begin{bmatrix}
{\boldsymbol 0}_{N-1} & -I_{N-1}
\end{bmatrix}',
$
and ${\boldsymbol 0}_{N-1}$ is an all-zeros vector
of size $N-1$.
The disagreement variable
${\boldsymbol z}(t)$ provides the errors between
a pivot agent and all other agents.
Such a transformation allows us to
reformulate the consensus
problem in Definition \ref{def:consensus} 
according to the following lemma:
\begin{lemma}[\cite{YuanGongSun2009}]
\label{lemma:disagreement}
The multi-agent 
system (\ref{eq:MA}), in closed-loop
with the control law (\ref{eq:input}), asymptotically
reaches the mean-square consensus if, and only
if, 
$
\lim_{t\rightarrow \infty}
\mathbb{E}\big(
\lVert{z}_i(t)\lVert^2
\big)\rightarrow 0,
$
for all $ i=1,\dots,N-1$.
\end{lemma}

We compute the equivalent multi-agent system by expressing (\ref{eq:MA}) in stacked form with the closed-loop control law (\ref{eq:input}) as follows,
\small
\begin{IEEEeqnarray}{ll}
\label{eq:MA_stacked}
{\dot {\boldsymbol x}}(t) =\ &\big(I_N \otimes A\big)
{\boldsymbol x}(t)
-(I_N\otimes B)\sat\big((\mathcal{L}_{\theta_t} \otimes K) 
{\boldsymbol x}(t)\big)
\nonumber \\
&+\big(I_{N} \otimes D \big)
{\boldsymbol w}(t)
,
\end{IEEEeqnarray}
\normalsize
where ${\w}(t) = [w_1(t)' ~\cdots~ 
w_{N}(t)']'$ and
$\mathcal{L}_{\theta_t}$ is a shorthand 
for
$\mathcal{L}{(\theta_t)}$ which 
represents the
Laplacian matrix associated to the current state of
the
Markov process $\{\theta_t:t\geq 0\}$. Taking
the time-derivative of (\ref{eq:disagreement}) 
and substituting ${\dot {\boldsymbol x}}(t)$ by
(\ref{eq:MA_stacked}) gives 
\small
\begin{IEEEeqnarray*}{rl}
{\dot {\boldsymbol z}}(t) =~&
(U\otimes I_m) \Big[ \big(I_N \otimes A\big)
{\boldsymbol x}(t)
-(I_N\otimes B)
\nonumber \\
&\times\sat\big((\mathcal{L}_{\theta_t}
\otimes K )
{\boldsymbol x}(t)
\big)
+\big(I_{N} \otimes D \big)
{\boldsymbol w}(t)
\Big]
,
\end{IEEEeqnarray*}
\normalsize
then, considering ${\boldsymbol x}(t)$ 
in (\ref{eq:MA_back_stacked}) and
applying the Kronecker identity 
$\big(M\otimes N)\big(P\otimes Q\big)=
\big(MP\otimes NQ)$, for matrices $M,~N,~P,$ and
$Q$ with appropriate dimensions, we have
\small
\begin{IEEEeqnarray*}{rll}
{\dot {\boldsymbol z}}(t) 
=~&\Bigr(
U{\boldsymbol 1}_{N}\otimes A{x}_{1}(t) 
+
(UW\otimes A){\boldsymbol z}(t)
\Bigr)
\nonumber \\
&-
(U\otimes B) \sat
\Bigr(
\big(\mathcal{L}_{\theta_t} {\boldsymbol 1}_{N} 
\otimes K{x}_{1}(t) \big)
+
(\mathcal{L}_{\theta_t}W\otimes K){\boldsymbol z}(t)
\Bigr)
\nonumber \\
&+\big(U \otimes D \big)
{\boldsymbol w}(t).
\end{IEEEeqnarray*}
\normalsize
Finally, 
noticing that
$U{\boldsymbol 1}_N=0$,
$\mathcal{L}_{\theta_t}{\boldsymbol 1}_N=0$,
and $UW=I_{N-1}$, we have
\small
\begin{IEEEeqnarray}{rl}
{\dot {\boldsymbol z}}(t) 
=~&
\big(
I_{N-1}\otimes A
\big){\boldsymbol z}(t)
-
\big(U\otimes B\big)
\sat\Bigr(\big(\mathcal{L}_{\theta_t}W\otimes K)
{\boldsymbol z}(t)\Bigr)
\nonumber \\
& +\big(U \otimes D \big)
{\boldsymbol w}(t)
.
\label{eq:stability}
\end{IEEEeqnarray}
\normalsize
Hence, we study the consensus problem of the multi-agent system
(\ref{eq:MA})
by analysing the stability of
(\ref{eq:stability}), where agents' coordinates are transformed to a relative coordinate system.

The non-linearity produced 
by the saturation satisfies a sector condition,
which allows us to analyze the dynamics of the multi-agent system as a Lur'e problem \cite{KHALIL2002ab}.
To do so, we can use the following dead-zone function
\cite{Tarbouriech2011},
\small
\begin{IEEEeqnarray}{ll}
\label{eq:phi_sat}
\Phi(\bu (t))= \bu (t) - \sat(\bu(t)),
\end{IEEEeqnarray}
\normalsize
with
$\bu(t) = (\mathcal{L}_{\theta_t}W\otimes K)
{\boldsymbol z}(t).$
Then, by summing and subtracting
$(U\otimes B)\bu(t)$ in equation
(\ref{eq:stability}) we get
\small
\begin{IEEEeqnarray}{rl}
{\dot {\boldsymbol z}}(t) 
=~&
\big(
I_{N-1}\otimes A
\big){\boldsymbol z}(t)
-
\big(U\mathcal{L}_{\theta_t}W\otimes BK)
{\boldsymbol z}(t) 
\label{eq:stab_sect}
\\
&+
\big(U\otimes B\big)
\Phi\Bigr(\big(\mathcal{L}_{\theta_t}W\otimes K)
{\boldsymbol z}(t)\Bigr)
+\big(U \otimes D \big)
{\boldsymbol w}(t).
\nonumber 
\end{IEEEeqnarray}
\normalsize
With the multi-agent system represented 
as in (\ref{eq:stab_sect}), we can define the 
following polyhedral set together with a sector
condition.
\begin{lemma}[Generalised Sector
Condition \cite{Tarbouriech2006}]
  For a given auxiliary signal $\boldsymbol{ \vartheta }(t)\in\mathbb{R}^{Np}$ and saturation
  limits $u_{\max}$, define the set
  \small
\begin{IEEEeqnarray}{ll}
 \mathbb{S}(\boldsymbol{\vartheta}(t),u_{\max})= \big\{&\boldsymbol{u}(t)\in \mathbb{R}^{Np}:
 \lvert (\boldsymbol{u}(t) 
 - \boldsymbol{\vartheta}(t))_{(q)}
 \lvert \leq u_{\max},
 \nonumber \\
 &\text{for } q = 1,...,Np\big\}.
\end{IEEEeqnarray}
\normalsize
If
${\boldsymbol{u}}(t)$
belongs to 
$\mathbb{S}(\boldsymbol{\vartheta}(t),u_{\max})$,
then
$\Phi \big( {\boldsymbol{u}}(t)\big)' T
\big[\Phi \big({\boldsymbol{u}}(t)) 
+ \boldsymbol{\vartheta}(t) \big] \leq 0 $
is
satisfied for any positive diagonal matrix 
$T\in \mathbb{R}^{Np\times Np}$,
where
$
\Phi\big( {\boldsymbol{u}}(t)\big)
$
is given by \eqref{eq:phi_sat}.
\label{lemma:gensec}
\end{lemma}

Since it might
not be possible to achieve
convergence from every 
initial conditions
for a system subject to
input saturation
\cite{Sussmann1994,Schmitendorf1980},
the state space under stabilizing linear feedback control can be partitioned into a region that leads to convergence and other that may lead to divergence. Therefore, we define the region of convergence as follows.
\begin{definition}
The region of convergence
is a subset of the state space,
$\mathscr{C}\subseteq
\mathbb{R}^{Nm}$,
such that for any initial conditions
starting
from $\mathscr{C}$, the 
multi-agent system without 
disturbances attain the mean-square
consensus. Namely,
\small
\begin{IEEEeqnarray}{ll}
\mathscr{C}=\big\{&\boldsymbol x(0)
\in\mathbb{R}^{Nm}:
\lim_{t\rightarrow \infty}
\mathbb{E}\big(
\lVert
x_i(t)-x_j(t)\lVert^2\big)\rightarrow 0,
\nonumber \\ 
&\forall
i,j\in \mathcal{V}\big\}.
\end{IEEEeqnarray}
\normalsize
\end{definition}
Exact characterization
the region of convergence
is a challenging,
even for single systems
\cite{Tarbouriech2011}.
We address this issue 
by defining a convex
optimization problem 
to compute 
a subset of 
$\mathscr{C}$ less
conservative as possible.

\section{Main Results}
\label{sec:3}
In this section, we derive stochastic stability criteria
for the disagreement system (\ref{eq:stability}).
In addition, we tackle the following problems:
i) the disturbance tolerance from the
consensus point,
ii) the disturbance tolerance
from arbitrary initial
conditions,
iii) the $\LLD_2$ gain,
and
iv) the maximization of the
region of convergence and the
region with guaranteed bounded states for 
the agents.

\begin{theorem}
\label{theorem:1}
Let the multi-agent system
(\ref{eq:MA}) in closed-loop with the consensus protocol
(\ref{eq:input})
subject to input saturation and stochastic
time-varying 
communication topology be represented by
equation (\ref{eq:stab_sect}).
For given positive scalars $\rho$, $\eta$, and
time-varying transition matrix $\Pi(t)$,
if there exist symmetric
positive definite
matrices 
$S\in\mathbb{R}^{Np\times Np}$,
$Y_{\ell}\in\mathbb{R}^{m(N-1)
\times m(N-1)}$,
and matrices 
$X_{\ell}\in\mathbb{R}^{Np\times m(N-1)}$
for all $\ell\in\{1,...,s\}$,
such that the following
matrix inequalities hold 
for all $i\in\{1,\dots,r\}$:
\small
\begin{IEEEeqnarray}{l}
\label{eq:main_theorem1}
\begin{bmatrix}
\Lambda & * & * & * \\
S\big(U\otimes B\big)'-X_{\ell} & -2S & * & * \\
\sqrt{\rho}(U \otimes D \big)' & 0 & 
\frac{1-\gamma}{\gamma}I & * \\
R_{\ell}' & 0 & 0 & -Q_{\ell}
\end{bmatrix}
\leq 0,
\end{IEEEeqnarray}
\begin{IEEEeqnarray}{ll}
\label{eq:inq_sat}
\begin{bmatrix}
Y_{\ell} & * \\
(\LL_{\ell}W\otimes K )_{(q)}Y_{\ell}-X_{\ell(q)} 
&
u_{\max}^2 \gamma
\end{bmatrix} \geq 0, 
\end{IEEEeqnarray}
for $q=1,...,Np,$
where  
\small
$Y_{\ell}=P_{\ell}^{-1}$,
\small
\begin{IEEEeqnarray*}{l} 
\Lambda=\mathrm{He}\Bigr(\big(
I_{N-1}\otimes A
-U\mathcal{L}_{\ell}W
\otimes BK\big)Y_{\ell}\Bigr)
+\pi_{\ell \ell}^i
Y_{\ell},
\\
R_{\ell}=\big[\sqrt{\pi_{\ell 1}^i}Y_{\ell}~\cdots~
 \sqrt{\pi_{\ell (\ell-1)}^i}Y_{\ell} 
 ~\sqrt{\pi_{\ell(\ell+1)}^i}Y_{\ell}
 ~\cdots~\sqrt{\pi_{\ell s}^i}Y_{\ell}\big],
 ~\text{and}
 \nonumber \\
Q_{\ell}=\diag\big( Y_{1},~\cdots,
 ~Y_{\ell-1},
 ~Y_{\ell+1},~\cdots
 ~Y_{s}\big),
\end{IEEEeqnarray*}
\normalsize
then:
\begin{itemize}
    \item[(i)] every trajectory of the multi-agent
system starting from the region
$\RR(\z(t),1)$
remains within
$\RR(\z(t),\gamma^{-1})$
for all $t\geq 0 $,
with 
$\gamma = 1/(1+N\rho \eta)$; 
    \item[(ii)] every trajectory of the multi-agent system
starting from the origin will remain within
the region
$\RR(\z(t),\gamma^{-1})$
for all $t\geq 0$,
with 
$\gamma = 1/(N\rho \eta)$; 
 and
\item[(iii)] in absence of disturbances,
$\w(t) = 0$,
the region
$\RR(\z(t),1)$
is an estimate included in 
the region of convergence $\mathscr{C}$,
and
the multi-agent system attains the mean-square
consensus asymptotically.
\end{itemize}
Moreover,
the set $\RR(\z(t),\sigma)$,
for constant $\sigma$,
is defined as
\small
\begin{IEEEeqnarray}{l}
\RR(\z(t),\sigma)=\bigcap_{\ell\in \cS}
\EE(P_\ell,\sigma),
\text{ with}
\nonumber \\
\EE(P_\ell,\sigma)=
\big\{\z(t)\in\mathbb{R}^{m(N-1)}:\z(t)'P_{\ell}\z(t)\leq
\sigma \big\}.
\quad
\end{IEEEeqnarray}
\normalsize
\end{theorem}
\begin{proof}
Consider the following
stochastic 
Lyapunov
candidate functional:
\small
\begin{IEEEeqnarray}{l}
\label{eq:Lfunction}
V({\z}(t),\theta_t=\ell)=\z(t)'
P_{\ell}\z(t), \text{ with }
\ell\in\cS,
\end{IEEEeqnarray}
\normalsize
where $P_{\ell}$ is a
constant positive definite matrix for each $\theta_t \in \cS$.
Let $\DD$ 
be the weak infinitesimal generator of the random
process $\{\theta_t:t\geq 0\}$.
Following
\cite{Ji1990,Bharucha-Reid1968},
the
difference of (\ref{eq:Lfunction}) along
the trajectories of (\ref{eq:stab_sect}) yields
\small
\begin{IEEEeqnarray*}{rl}
\DD V({\z}(t))
=~&
\lim_{\Delta\rightarrow 0}
\frac{1}{\Delta}
\left[
V({\z}(t+\Delta),\theta_{t+\Delta})
-V({\z}(t),\theta_t=\ell)
\right]
\nonumber \\
=~&
2{\dot \z}(t)'P_\ell \z(t) 
+\z(t)'\sum_{j=1}^s\pi_{\ell j}(\alpha(t))
P_{j}\z(t)
\nonumber \\
=~& 2{\z}(t)'P_\ell\Big[
\big(
I_{N-1}\otimes A
-
U\mathcal{L}_{\ell}W\otimes BK\big)
{\boldsymbol z}(t) 
\nonumber \\
&+
\big(U\otimes B\big)
\Phi\Bigr(\big(\mathcal{L}_{\ell}W\otimes K)
{\boldsymbol z}(t)\Bigr)
+\big(U \otimes D \big)
{\boldsymbol w}(t) \Big]
\nonumber \\
&+\z(t)'\sum_{j=1}^s\pi_{\ell j}(\alpha(t))
P_{j}\z(t),
\end{IEEEeqnarray*}
\normalsize
noting that
\small
$
2{\z}(t)'P_\ell\big (U \otimes D  \big)\w(t)
\leq 
\frac{1}{\eta}{\z}(t)'P_\ell\big(U \otimes D \big)
 (U \otimes D \big)'P_\ell\z(t)
+  \eta  \w(t)'\w(t),
$
\normalsize
with $\eta>0$, we have
\small
\begin{IEEEeqnarray}{ll}
\label{eq:Lyap_deriv}
\DD V(\z(t))
\leq~& 2{\z}(t)'P_\ell\Big[
\big(
I_{N-1}\otimes A
-
U\mathcal{L}_{\ell}W\otimes BK\big)
{\boldsymbol z}(t) 
\\
&+
\big(U\otimes B\big)
\Phi\Bigr(\big(\mathcal{L}_{\ell}W\otimes K)
{\boldsymbol z}(t)\Bigr) 
\Big]
\nonumber \\
&
+\frac{1}{\eta}{\z}(t)'P_\ell\big(U \otimes D \big)
 (U \otimes D \big)'P_\ell\z(t)
 \nonumber \\
&+ \eta  \w(t)'\w(t)
+\z(t)'\sum_{j=1}^s\pi_{\ell j}(\alpha(t))
P_{j}(\alpha(t))\z(t).
\nonumber 
\end{IEEEeqnarray}
\normalsize

For
the sake of continuity,
we assume that $\bu(t)\in 
\mathbb{S}(\boldsymbol{\vartheta}(t),u_{\max})$,
a sufficient condition for this
assumption is given in the
sequence.
As a result, the inequality 
$
-2\Phi \big( {\boldsymbol{u}}(t)\big)' T
\big[\Phi \big({\boldsymbol{u}}(t)) 
+ \boldsymbol{\vartheta}(t) \big] \geq 0
$
holds, in accordance with Lemma 
\ref{lemma:gensec}.
Considering the auxiliary signal
${\boldsymbol \vartheta}(t)=G\z(t)$ and adding the
previous inequality to equation 
\eqref{eq:Lyap_deriv} gives,
\small
\begin{IEEEeqnarray}{ll}
\label{eq:aux_theorem2}
\DD V(\z(t))
\leq~& \chi(t)'\Psi\chi(t)+\eta\w(t)'\w(t),
\end{IEEEeqnarray}
\normalsize
with $\chi(t)' =
[\z(t)'~\Phi(\bu(t))']$ and
\small
\begin{IEEEeqnarray*}{l}
\Psi=
\begin{bmatrix}
\Omega^1 & * \\
\big(U\otimes B\big)'P_{\ell}-TG & -2T
\end{bmatrix}.
\end{IEEEeqnarray*}
\normalsize
With \small
$
\Omega^1 = \mathrm{He}\Bigr(P_{\ell}\big(
I_{N-1}\otimes A
-U\mathcal{L}_{\ell}W\otimes BK\big)\Bigr)
+\frac{1}{\eta}P_\ell\big(U \otimes D \big)
 (U \otimes D \big)'P_\ell
 +\sum_{j=1}^s\pi_{\ell j}(\alpha(t))
P_{j}(\alpha(t)).
$ 
\normalsize
By left- and right-multiplying $\Psi$ by
$\diag(P_{\ell}^{-1},T^{-1})$, making
the change of variables $Y_{\ell}=P_{\ell}^{-1}$,
$S=T^{-1}$, $X_{\ell}=GP_{\ell}^{-1}$, 
also noticing from \eqref{eq:polytope}
that
$\pi_{\ell j}(\alpha(t))=\sum_{i=1}^r\alpha_i\pi_{\ell j}^i$,
it is sufficient to test the conditions on the
vertices of $\Pi(t)$ \cite{Geromel1991}, then
applying Schur complement two consecutive times 
on the non-linear terms gives 
inequality \eqref{eq:main_theorem1}. 
Therefore, if inequality \eqref{eq:main_theorem1} 
is satisfied, we
have that 
$
\DD V(\z(t))\leq \eta \w(t)'\w(t).
$
From Dynkin's formula, we
have that for all
$\ell\in\cS$,
\small
\begin{IEEEeqnarray*}{rl}
\mathbb{E}\big(V(\z(t)) \big)
\leq~ &
V(\z(0))+ \eta\mathbb{E}\left(\int_0^t\w(\tau)'\w(\tau)d\tau
\right)
\nonumber \\
< ~ & V(\z(0))+N\eta \rho.
\end{IEEEeqnarray*}
\normalsize
Finally, we establish the three cases 
in the Theorem \ref{theorem:1}, namely:
(i) if for each $\ell\in\cS$
we have
$V(\z(0))\leq 1$, that is
$\z(0)\in\EE(P_{\ell},1)$, all
trajectories will remain within
$\EE(P_\ell,1+N\rho\eta)=
\big\{\z(t)\in\mathbb{R}^{m(N-1)}:\z(t)'P_\ell\z(t)\leq
1+N\rho\eta\big\}$
for all $\ell\in\cS$
and $t\geq 0$;
(ii) if, for each $\ell\in\cS$,
$V(\z(0)) = 0$, which implies
that $\z(0) = 0$, then 
$\z(t)\in
\EE(P_\ell,N\rho\eta)=
\big\{\z(t)\in\mathbb{R}^{m(N-1)}:
\z(t)'P_\ell\z(t)\leq
N\rho\eta\big\},~\forall \ell\in\cS$, and finally;
(iii) in absence of disturbances 
$\DD V(\z(t))\leq 0$, and $V(\z(t))=0$ if, and
only if $\z(t)=0$,
whenever for each 
$\ell\in\cS$,
$\z(0)\in\EE(P_\ell,1)$. 
Hence, 
the set $\EE(P_\ell,1)$ is 
an estimate included in the region of convergence.

Henceforth, we demonstrate that if inequality
\eqref{eq:inq_sat}
is satisfied, then our previous assumption
that
$\bu(t)\in 
\mathbb{S}(\boldsymbol{\vartheta}(t),u_{\max})$ holds.
Replacing $Y_{\ell}$ and $X_{\ell}$ by 
$P_{\ell}^{-1}$ and $P_{\ell}^{-1}G'$ 
in 
\eqref{eq:inq_sat}, respectively,
left- and right-multiplying the result
by $\diag(P_{\ell},1)$, and recalling
that $\gamma$ has a different
description according the
cases (i)-(iii), we get
\begin{IEEEeqnarray*}{ll}
\begin{bmatrix}
P_{\ell} & * \\
(\LL_{\ell}W\otimes K )_{(q)}-G_{(q)} &
\gamma u_{\max}^2
\end{bmatrix} \geq 0.
\end{IEEEeqnarray*}
By applying Schur complement and left- and
right-multiplying the result by
$\z(t)'$ and $\z(t)$, respectively, we get that
\small
\begin{IEEEeqnarray*}{l}
\z(t)' \Bigr((\LL_{\ell}W\otimes
K )_{(q)}-G_{(q)}\Bigr)'\gamma^{-1}u_{\max}^{-2}\Bigr((\LL_{\ell}W\otimes
K )_{(q)}-G_{(q)}\Bigr)
\nonumber \\ 
\times \z(t)
\leq
\z(t)'P_{\ell}\z(t)\leq V(\z(t)).
\end{IEEEeqnarray*}
\normalsize
Hence,
condition \eqref{eq:inq_sat} ensures that
$\bu(t)\in\mathbb{S}({\boldsymbol \vartheta},u_{\max})$, and therefore the
inequality in Lemma \ref{lemma:gensec}
is satisfied.
\end{proof}

Theorem \ref{theorem:1} establishes conditions for bounded trajectories and mean-square consensus in networked systems,
according to agents disturbances.
However, important questions remain unanswered,
including identifying the largest disturbance bounds for bounded trajectories, minimizing the difference between estimated and actual convergence regions, and estimating the $\LLD_2$ gain on $\WW$.
Addressing these problems is crucial for designing and deploying multi-agent systems. The following sections will discuss each of these issues.

\subsection{Optimization methods for design}
\subsubsection{Maximization of 
the disturbance tolerance}
Maximizing the disturbance tolerance involves maximizing value of $\rho$, 
such that
computing estimates for the convergence set and the set that contains the trajectories of the network is feasible.
We formulate this as the following optimization problem:
\small
\begin{IEEEeqnarray}{l}
\label{opt:disturbancies}
\left\{
\begin{matrix}
\underset{{\gamma\in(0,1)}}{\sup} & \bar \rho
\hfill \\
\text{s.t.} &
(a)~\EE(Z,1)
\subseteq
\EE(P_{\ell},1), 
\hfill
\\
&
(b)~\text{LMIs \eqref{eq:main_theorem1}, \eqref{eq:inq_sat}},
\hfill
\end{matrix}
\right.
\end{IEEEeqnarray}
\normalsize
for all $\ell\in\{1,\dots,s\}$, with 
${\bar \rho}=\sqrt{\rho}$. The constraint
$(a)$
corresponds to placing
an ellipsoid inside the
intersection of ellipsoids
defined by $P_{\ell}$ for all 
$\ell\in\{1,\dots,s\}$, and
it
is equivalent to
\begin{IEEEeqnarray*}{l}
\begin{bmatrix}
Z & *\\ I & Y_{\ell} 
\end{bmatrix}\geq0,
\text{ for $\ell = 1,\dots,s$,}
\end{IEEEeqnarray*}
with $Y_{\ell}=P_{\ell}^{-1}$.
Clearly, all 
constraints in 
\eqref{opt:disturbancies} are
linear matrix inequalities for
fixed values of $\gamma\in(0,1)$.

Another relevant issue is 
finding the maximal disturbance
tolerance when the
initial condition is
the equilibrium point,
\textit{i.e.}, $\z(0) = 0$. 
This value
is critical to 
estimating the $\LLD_2$ gain
of the network.
This can be cast
similarly as the 
optimization problem
\eqref{opt:disturbancies},
by letting $\eta=1$ and
$\gamma = 1/N\rho$.
By applying Schur complement 
on inequality \eqref{eq:main_theorem1}
we get,
\small
\begin{IEEEeqnarray}{l}
\label{eq:opt_rho_zero}
\begin{bmatrix}
 \Lambda + (U \otimes D \big)(U \otimes D \big)'
& * & * \\
S\big(U\otimes B\big)'-X_{\ell} & -2S & * \\
R_{\ell}' & 0 & -Q_{\ell}
\end{bmatrix}
\leq 0.
\end{IEEEeqnarray}
\normalsize
Therefore, the problem of estimating the
maximal disturbance tolerance from the consensus point 
can be cast as,
\small
\begin{IEEEeqnarray}{l}
\label{opt:disturbancies_zero}
\left\{
\begin{matrix}
{\min} & \gamma 
\hfill \\
\text{s.t.} &
(a)~\EE(Z,1)\subseteq
\EE(P_{\ell},1),
\hfill
\\
&
(b)~\text{LMIs \eqref{eq:opt_rho_zero}, \eqref{eq:inq_sat}},
\hfill
\end{matrix}
\right.
\end{IEEEeqnarray}
\normalsize
for all $\ell\in\{1,\dots,s\}$, 
and $\rho$ is
given by ${\rho }={1}/{N\gamma}$.

\subsubsection{Estimation of the
$\LLD_2$ gain}
\label{subsubsec:estimationL2}
The $\LLD_2$ gain represents 
the ratio between
a bounded output
and a bounded
input for a system 
starting from the origin.
In our scenario, we define an output for the multi-agent system 
\eqref{eq:MA} on the disagreement variable as,
\small
$
\y(t) = C\z(t),
$
\normalsize
in which $C$ is an matrix with
appropriate dimension that
properly weight disagreement variables of interest.
Therefore, we compute the $\LLD_2$ gain of the multi-agent system,
subject to
energy bounded disturbances according the following
corollary:
\begin{coro}[$\boldsymbol{\LLD_2}$
gain]
For given positive scalars $\rho$
and $\varrho$, and
time-varying transition matrix $\Pi(t)$,
if there exist symmetric
positive definite
matrices 
$S\in\mathbb{R}^{Np\times Np}$,
$Y_{\ell}\in\mathbb{R}^{m(N-1)
\times m(N-1)}$,
and matrices 
$X_{\ell}\in\mathbb{R}^{Np\times m(N-1)}$
for all $\ell\in\{1,...,s\}$,
such that the following
matrix inequalities hold
for all $i\in\{1,\dots,r\}$:
\small
\begin{IEEEeqnarray}{l}
\label{eq:main_theorem2}
\begin{bmatrix}
\Lambda +
(U \otimes D \big)(U \otimes D \big)' & * & * & * \\
S\big(U\otimes B\big)'-X_{\ell} & -2S & * & *
\\
R_{\ell}' & 0 & -Q_{\ell} & * \\
CY_{\ell} & 0 & 0 & -\varrho^2 I
\end{bmatrix}
\leq 0,
\end{IEEEeqnarray}
\begin{IEEEeqnarray}{ll}
\label{eq:inq_sat2}
\begin{bmatrix}
Y_{\ell} & * \\
(\LL_{\ell}W\otimes K )_{(q)}-X_{\ell(q)} &
u_{\max}^2 \gamma
\end{bmatrix} \geq 0,
\end{IEEEeqnarray}
\normalsize
 for $q=1,...,Np,$
 where $\gamma = 1/N\rho$,  
$Y_{\ell}=P_{\ell}^{-1}$,
and $\Lambda$, $R_{\ell}$, and $Q_{\ell}$ are
defined in \ref{eq:main_theorem1}.
Then, the $\LLD_2$ gain from 
$\w(t)$ to $\y(t)$ for all
$w_i(t)\in \mathcal{W}$ is no larger 
than $\varrho$.
\end{coro}

\begin{proof}
Similarly to Theorem \ref{theorem:1},
inequality
\eqref{eq:inq_sat2} ensures 
that
$\bu(t)\in 
\mathbb{S}(\boldsymbol{\vartheta}(t),u_{\max})$ whenever
$\z(t)\in 
\EE(P_{\ell},N\rho)$.
We have that \eqref{eq:main_theorem2}
is equivalent to
\small
\begin{IEEEeqnarray}{ll}
\label{eq:L2_equiv}
~& 2{\z}(t)'P_\ell\Big[
\big(
I_{N-1}\otimes A
-
U\mathcal{L}_{\ell}W\otimes BK\big)
{\boldsymbol z}(t) 
\nonumber \\
&+
\big(U\otimes B\big)
\Phi\bigr(\bu(t)\bigr) 
\Big]
+{\z}(t)'P_\ell\big(U \otimes D \big)
 (U \otimes D \big)'P_\ell\z(t)
\nonumber \\
&
+\frac{1}{\varrho^2}
{\z}(t)'C'C\z(t)
+\z(t)'\sum_{j=1}^s\pi_{\ell j}(\alpha(t))
P_{j}(\alpha(t))\z(t)
\nonumber \\
&-2\Phi \big( {\boldsymbol{u}}(t)\big)' T
\big[\Phi \big({\boldsymbol{u}}(t)) 
+ G\z(t) \big]
\leq 0,
\end{IEEEeqnarray}
\normalsize
which can be shown by applying Schur complement on
\eqref{eq:main_theorem2} two times  
consecutively,
making the change of variables
$Y_{\ell}=P_{\ell}^{-1}$, $S=T^{-1}$, 
$X_{\ell}=GP_{\ell}^{-1}$, left- and 
right-multiplying 
it by $\diag(P_{\ell},T)$
and by $[\z(t)'~\Phi(\bu(t))']$
and its transpose, respectively.
Notice that this condition
contains inequality
\eqref{eq:main_theorem1} with
$\eta=1$. Therefore, 
admitting that
\eqref{eq:main_theorem2} is
satisfied and considering
inequalities 
\eqref{eq:aux_theorem2} 
and \eqref{eq:L2_equiv} we
have that
\small
\begin{IEEEeqnarray*}{ll}
\DD V(\z(t))
\leq~& -\frac{1}{\varrho^2}
\z(t)'C'C\z(t)+\w(t)'\w(t),
\end{IEEEeqnarray*}
\normalsize
which from Dynkin's formula
yields
\small
\begin{IEEEeqnarray*}{lll}
 \mathbb{E}\big(V(&\z(t))\big)
\leq 
\nonumber \\
&-\frac{1}{\varrho^2}
\mathbb{E}\Bigg(
\int_0^t \z(\tau)'C'C\z(\tau)d\tau
\Bigg)
+
\mathbb{E}\Bigg(
\int_0^t\w(\tau)'\w(\tau)d\tau
\Bigg).
\end{IEEEeqnarray*}
\normalsize
Noticing that $V(\z(t))\geq0$ and
$w_i(t)\in\WW$, for all $i\in\mathcal{V}$, 
implies
\small
\begin{IEEEeqnarray*}{ll}
\mathbb{E}\Bigg(
\int_0^t \y(\tau)'\y(\tau)d\tau 
 \Bigg)<
\varrho^2
N\rho,
\end{IEEEeqnarray*}
\normalsize
which concludes the
demonstration.
\end{proof}

\subsubsection{Synthesis of the controllers' gains}
The problem of synthesizing the gains of the feedback matrices
can be tackled by letting 
the matrix $K$ be an additional variable
on the previous formulations.
We provide this procedure as a corollary,
since a particular choice for the
structure of
some variables is necessary to produce
linear convex problems.
The adaptation of Theorem \ref{theorem:1}
to compute the feedback gains
is given as follows.

\begin{coro}
\label{coro:1}
Let the networked system
(\ref{eq:MA}) in closed-loop with the consensus protocol
(\ref{eq:input})
subject to input saturation and stochastic
time-varying 
communication topology be represented by
equation (\ref{eq:stab_sect}).
For given positive scalars $\rho$ and $\eta$,
time-varying transition matrix $\Pi(t)$,
if there exist symmetric
positive definite
matrices 
$S\in\mathbb{R}^{Np\times Np}$,
$\bar Y_{\ell}=I_{N-1}\otimes 
F\in\mathbb{R}^{m(N-1)
\times m(N-1)}$,
and matrices 
$X_{\ell}\in\mathbb{R}^{Np\times m(N-1)}$
for all $\ell\in\{1,...,s\}$,
such that the following
matrix inequalities hold
for all $i\in\{1,\dots,r\}$:
\small
\begin{IEEEeqnarray}{l}
\label{eq:synt_main}
\begin{bmatrix}
\bar{\Lambda} & * & * & * \\
S\big(U\otimes B\big)'-X_{\ell} & -2S & * & * \\
\sqrt{\rho}(U \otimes D \big)' & 0 & 
\frac{1-\gamma}{\gamma}I & * \\
R_{\ell}' & 0 & 0 & -Q_{\ell}
\end{bmatrix}
\leq 0,
\end{IEEEeqnarray}
\begin{IEEEeqnarray}{ll}
\label{eq:inq_sat_synth}
\begin{bmatrix}
\bar Y_{\ell} & * \\
(\LL_{\ell}W\otimes \bar K )_{(q)}-X_{\ell(q)} 
&
u_{\max}^2 \gamma
\end{bmatrix} \geq 0, \text{for } q=1,...,Np,
\quad
\end{IEEEeqnarray}
where,  
\begin{IEEEeqnarray*}{l} 
\bar{\Lambda}=\mathrm{He}\Bigr(\big(
I_{N-1}\otimes A\big)\bar Y_{\ell}
-U\mathcal{L}_{\ell}W
\otimes B\bar K\Bigr)
+\pi_{\ell \ell}^i
\bar Y_{\ell}\\
R_{\ell}=\big[\sqrt{\pi_{\ell 1}^i}\bar Y_{\ell}~\cdots~
 \sqrt{\pi_{\ell (\ell-1)}^i}\bar Y_{\ell} 
 ~\sqrt{\pi_{\ell(\ell+1)}^i}\bar Y_{\ell}
 \cdots\sqrt{\pi_{\ell s}^i}\bar Y_{\ell}\big],
 \text{ and}
 \nonumber \\
Q_{\ell}=\diag\big( \bar Y_{1},~\cdots,
 ~\bar Y_{\ell-1},
 ~\bar Y_{\ell+1},
 ~\cdots,
 ~\bar Y_{s}\big),
\end{IEEEeqnarray*}
\normalsize
then items (i)-(iii) in Theorem \ref{theorem:1}
hold
with agents in closed-loop with
feedback matrix $K=\bar{K}F^{-1}$.
\end{coro}
\begin{proof}
The demonstration follows the same lines
of
Theorem \ref{theorem:1}.
Replacing the matrices
variables $Y_{\ell}$ by the
structured variables
$\bar Y_{\ell}=I_{N-1}\otimes F$
for all $\ell\in\cS$. 
Then, 
the conditions 
are obtained by
applying the Kronecker identity
in all products with the
gain matrix $K$
and, subsequently, making the change of 
variables $\bar K=KF$.
\end{proof}

\begin{remark}
\label{remark_comp_burden}
The conditions proposed in 
Theorem \ref{theorem:1}
and Corollary 
\ref{coro:1}
are Linear Matrix Inequalities, hence
their computational complexity 
grow with the number of decision variables
\cite{Boyd1994}.
In our scenario, the
number of agents, the 
dimension of their state-space, as well
as the number of vertices of the convex hull
of the
the time-varying transition matrix
determine the
computational complexity of the conditions.
Precisely, we have
$\frac{\ell}{2}\big((Np)^2 
+ (m(N-1))^2
+2N^2pm
+ Np  + m(N-1)
 - 2Npm
\big)$
variables in the conditions of Theorem
\ref{theorem:1},
and 
$\frac{\ell}{2}\big((Np)^2 
+2N^2pm
+\frac{m^2+m}{\ell} 
 - 2Npm
\big)$
variables in the conditions of
Corollary \ref{coro:1}.
Although these numbers might be high
for some systems, 
the proposed results
achieve less conservative results in
a more general context than
similar works from the literature. In addition,
all conditions are computed off-line.
\end{remark}

\section{Simulation Results}
\label{sec:4}
\begin{example}
We consider the problem
of maximizing the disturbance
tolerance of a multi-agent
system starting from the 
equilibrium point.
The network is described by
(\ref{eq:MA}) with
\small
\begin{IEEEeqnarray*}{ll}
A = \begin{bmatrix}
0.1 & -0.1 \\ 0.1 & -3.0
\end{bmatrix},
~
B = \begin{bmatrix}
5 & 0\\ 0 & 1
\end{bmatrix},
~
K = \begin{bmatrix}
0.1 & 0.5 \\ 0 & 0
\end{bmatrix},
\end{IEEEeqnarray*}
\normalsize
and $D= I_2$, where $K$ is computed using Corollary 3.3 and then fixed.
In addition, we assume
that the switching of the 
time-varying communication
topology is captured by the
Markov process defined by the following
transition matrix:
\small
\begin{IEEEeqnarray*}{l}
\Pi = \begin{bmatrix}
-2 & 1 & 1\\
2 & -4 & 2\\
1 & 1 & -2
\end{bmatrix},
\end{IEEEeqnarray*}
\normalsize
and the topologies associated with
each state of the Markov process
are given by
\small
\begin{IEEEeqnarray*}{l}
\LL_{1} = \begin{bmatrix}
1 & 0 & -1\\
0 & 0 & 0\\
0 & -1 & 1
\end{bmatrix};
\LL_{2} = \begin{bmatrix}
0 & 0 & 0\\
0 & 0 & 0\\
0 & 0 & 0
\end{bmatrix};
\LL_{3} = \begin{bmatrix}
0 & 0 & 0\\
0 & 1 & -1\\
0 & 0 & 0
\end{bmatrix}.
\end{IEEEeqnarray*}
\normalsize
Notice that no single topology is connected, however, the union of the graphs
has a directed spanning tree.

By setting a limit for the actuators and solving optimization problem
(\ref{opt:disturbancies_zero})
for the multi-agent system starting from the equilibrium point, we obtain the maximum energy bound, 
$N\rho$, for disturbances affecting the network shown in 
Figure \ref{fig:satXdist}.
Unsurprisingly,
this indicates that the ability of
disturbance rejection of the network 
grows with the
actuation limit. 
\begin{figure}[h]
\centering
\includegraphics[scale=0.5]{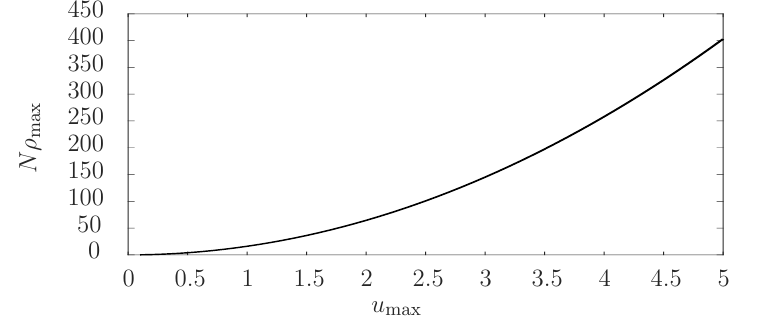}
\caption{\footnotesize{
Variation of disturbance level
rejection according saturation limit.}}
\label{fig:satXdist}
\end{figure}

For $u_{\max}=3$ we
have $N\rho=145.1$, hence if the
energy of the disturbances that
impacts the network is less than
$145.1$, the agents' states
are guaranteed to 
remain within the 
designed region. 
\begin{figure}[h]
\centering
\includegraphics[scale=0.73]{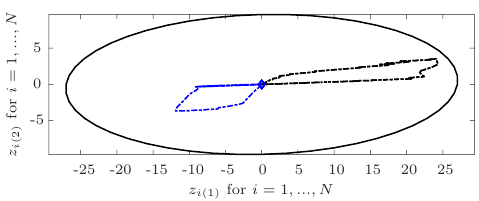}
\caption{\footnotesize{Two realizations
of the multi-agent system trajectories
subject to ramp (in black
line) and constant (in blue line) 
disturbances with energy limited by 145.}}
\label{fig:traj_ex1}
\end{figure}

Figure \ref{fig:traj_ex1} depicts a realization
of the trajectories
of the multi-agent system from the
equilibrium point subject to 
a ramp disturbance 
impacting only agent $1$ in black line, \textit{i.e.}, 
$w_1(t) =[t~t]'$ and $w_i(t) = [0~0]'$ for
$i=2,3,$ and a constant disturbance
impacting only agent $2$ in blue line, \textit{i.e.},
$w_2(t) =[10~10]'$ and $w_i(t) = [0~0]'$ for
$i=1,3,$ both with specific duration
such that $N\rho_{\max}\leq 145$.
The control inputs and switching topologies
are shown in Figure 
\ref{fig:swit_ctrl} for both cases.
\begin{figure}[ht]
\centering
\includegraphics[scale=0.50]{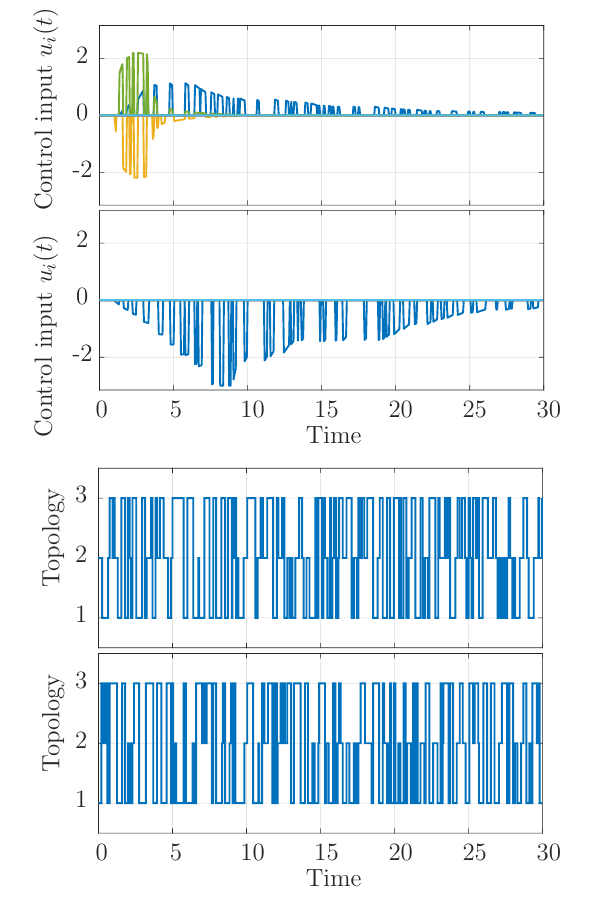}
\caption{\footnotesize{Control
signals and switching topologies for
the multi-agent system
subject to ramp (in the bottom)
and constant (in the top) 
disturbances with energy limited by 145.}}
\label{fig:swit_ctrl}
\end{figure}
\label{example:1}
\end{example}

\begin{example}
We show the impact of the
switching frequency of the
stochastic communication on the convergence region size. Using the same
system as
in Example \ref{example:1},
we optimize
the region $\RR(\z(t),1)$ for different 
transition matrices, according to case (i) of
Theorem \ref{theorem:1}.
 By replacing the
objective function in
(\ref{opt:disturbancies})
by
$\min{~\mathrm{Trace}(Z)}$
and setting $\gamma = 0.8$,
we solve optimization problem
with $\rho$ as
variable. 
Figure \ref{fig:traceXepsi} 
shows the
trace of $Z$ as function
of $\varepsilon$, a positive
constant that scales the
transition
matrix as follow
\small
\begin{IEEEeqnarray*}{l}
\Pi = \varepsilon
\begin{bmatrix}
-2 & 1 & 1\\
2 & -4 & 2\\
1 & 1 & -2
\end{bmatrix}.
\end{IEEEeqnarray*}
\normalsize
Figure \ref{fig:traceXepsi} 
shows  
a non-linear relation 
between the size of of the
estimate of the region of 
convergence and the 
probability of link formation.
Decreasing the
frequency of interactions reduces the regions of 
convergence, however the
biggest region for this network
occurs when $\varepsilon=0.51$.
For larger values the size starts to
contract.
\begin{figure}[h]
\centering
\includegraphics[scale=0.45]{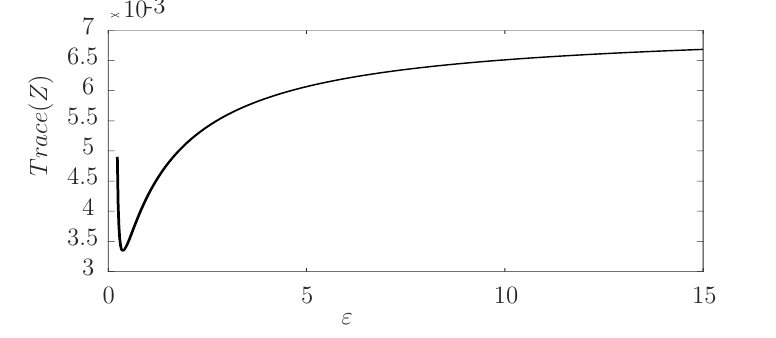}
\caption{\footnotesize{
Trace of $Z$ as function
of different scales
parameters of the
transition matrix $\Pi$.}}
\label{fig:traceXepsi}
\end{figure}
The estimate of the region of
convergence are
shown in Figure \ref{fig:tree_ell} for
$\varepsilon=0.25$,
$\varepsilon=0.51$,
and
$\varepsilon=15.00$.

\begin{figure}[h]
\centering
\includegraphics[scale=0.55]{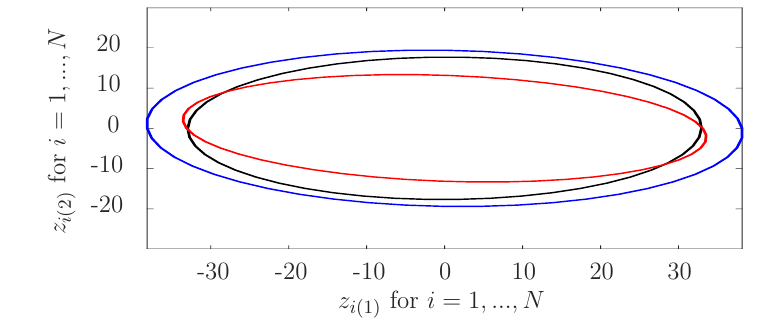}
\caption{\footnotesize{
Estimates of the region of
convergence for three
different values of
$\varepsilon$, in black
for $\varepsilon=0.25$,
in blue for 
$\varepsilon=0.51$, and
in $\varepsilon=15.00$ in 
red line.}}
\label{fig:tree_ell}
\end{figure}

\end{example}

\section{Conclusions}
\label{sec:5}
In this paper, we investigated the effects of disturbances and stochastic intermittent communications on agents performing consensus under saturating inputs.
We formulated the problem as a stability problem of Markov jump linear systems, proposed conditions to estimate the region of convergence and the region that bounds the trajectories of the agents' states, and demonstrated the formulation of optimization convex methods to design networks with optimized parameters.
Numerical simulations illustrate our proposed approach.
Future work includes finding specific control gains for different network topologies, formulating the consensus problem for heterogeneous networks, and investigating event-triggered control approaches.


%
\bibliographystyle{elsarticle-num}
\bibliography{references.bib}


\end{document}